\documentclass[11pt]{article}
\usepackage{graphicx}
\usepackage{epstopdf}
\usepackage{relsize,exscale}
\usepackage{amssymb}
\usepackage{bm}
\usepackage{amsmath}
\usepackage{cancel}
\usepackage[margin=1.0in]{geometry}
\usepackage{cite}
\usepackage[shortlabels]{enumitem}
\usepackage{scalerel}
\usepackage{color}
\usepackage{enumitem}
\usepackage[titletoc]{appendix}
\usepackage[font=small,labelfont=bf]{caption}
\DeclareRobustCommand{\stirling}{\genfrac\{\}{0pt}{}}

\allowdisplaybreaks

\usepackage{amsthm}
\newtheorem{theorem}{Theorem}

\newtheorem{lemma}{Lemma}

\newtheorem{conjecture}{Conjecture}
\usepackage{chngcntr}
\counterwithout{theorem}{section}
\counterwithout{definition}{section}
\counterwithout{lemma}{section}
\counterwithout{remark}{section}
\counterwithout{assumption}{section}
\counterwithout{proposition}{section}
\counterwithout{corollary}{section}
\counterwithout{claim}{section}
\counterwithout{conjecture}{section}
\begin{document}
\title{Collusion-resistant fingerprinting of parallel content channels}
\author{Basheer Joudeh, Boris \v{S}kori\'{c}}
\date{}
\maketitle

\abstract{
The fingerprinting game is analysed when the coalition size $k$ is known to the tracer, but the colluders can distribute themselves across $L$ TV channels. The collusion channel is introduced and the extra degrees of freedom for the coalition are made manifest in our formulation. We introduce a payoff functional that is analogous to the single TV channel case, and is conjectured to be closely related to the fingerprinting capacity. For the binary alphabet case under the marking assumption, and the restriction of access to one TV channel per person per segment, we derive the asymptotic behavior of the payoff functional. We find that the value of the maximin game for our payoff is asymptotically equal to $L^2/k^2 2 \ln 2$, with optimal strategy for the tracer being the arcsine distribution, and for the coalition being the interleaving attack across all TV channels, as well as assigning an equal number of colluders across the $L$ TV channels.
}

\section{Introduction}

\subsection{Collusion resistant fingerprinting}

Fingerprinting, also known as
forensic watermarking, is a technique for tracing the origin
and distribution of digital content.
Before distribution, the content is modified by embedding an imperceptible watermark,
which is unique for each recipient.
When an unauthorized copy of the
content is released, the watermark in this copy reveals information about the
identities of those who created the copy.
A tracing algorithm (also called a decoder) outputs a list of suspicious users.
This procedure is known as forensic watermarking or traitor tracing.

The most powerful attack against watermarking is the {\em collusion
attack}, where multiple users (the `coalition') combine their
differently watermarked versions of the same content; the detected
differences partly reveal the locations of the hidden marks and allow for an informed attack.
Various collusion-resistant codes have
been developed, most notably the
class of {\em bias-based} codes, introduced by G.\,Tardos in
2003 \cite{Tardos,Tardos2008}.
For each watermarking symbol position the tracer first generates a bias $w\in(0,1)$
drawn from a distribution $f_W$
and then assigns to each user a watermark symbol `1' with probability $w$ and `0' with probability $1-w$.
Work on bias-based codes includes
improved analyses
\cite{BlayerTassa,Furon08,Furon_length,LaarhovenWeger,TardosFourier,Vladimirova,revisited},
code modifications
\cite{HuangMoulinWIFS09,Nuida_c3,Nuida_DCC09}, advanced decoders
\cite{AmiriTardos,FuronEM,DonQuixote,OosterwijkIH2013,FuronDesoubeaux,Skoric2015,SdG2015} and generalizations
\cite{Furon_asym,symmetric,CDM,XFF}.
Bias-based codes
achieve the asymptotically optimal relationship $n \propto k^2$, where
$n$ is the  sufficient code length, and $k$ is the coalition size.

An important result was finding the asymptotic {\em saddlepoint}
of the information-theoretic maximin game \cite{HuangMoulinWIFS09,HM2012ISIT,HM2012TIFS}
in the case of the Restricted Digit Model\footnote{
In the Restricted Digit Model, colluders must output a symbol that has been received by at least one of them.
}
and joint decoding.

The saddlepoint is a pair (bias distribution $f_W$, attack strategy) such that it is disadvantageous for either party to
depart from their strategy.
With increasing $k$, the solution of the max-min game for the binary fingerprinting alphabet
gets closer to the combination\footnote{
This $f_W$ is known as the arcsine distribution, $f_W(w){\rm d} w = \frac{1}{\pi}{\rm d} \arcsin(2w-1)$.
}
($f_W=\frac1{\pi\sqrt{w(1-w)}}$
, attack = Interleaving).
In the Interleaving attack the colluders output the symbol of one colluder chosen uniformly at random.

Knowing the location of the saddlepoint allows the tracer to build
a {\em universal} decoder that works optimally against the saddlepoint attack and
that works well against all other attacks too. What is usually not considered in studies of forensic watermarking is that
most pirate decoder boxes observed in practice give access to multiple TV channels in parallel. Hence, attackers have an additional degree of freedom that has not yet been explored in the academic literature:
which TV channel to collude on at which point in time. The information-theoretic maximin game has not yet been studied for multiple TV channels attacks.


\subsection{Contributions and outline}

We study the information-theoretic maximin game for the binary fingerprinting scenario with multiple parallel TV channels
which are being attacked simultaneously by a set of colluders under the Restricted Digit Model.
We consider the {\em static} case, as opposed to {\em dynamic} traitor tracing, i.e.\;we do not allow the parties to adapt their strategy
as a function of symbols observed previously.
We assume that each attacker can tune into merely one channel, and furthermore we consider only attack strategies in which
the colluders take equal risk, and all TV channels are treated as being equally important.

Under these restrictions we study the mutual information $I(\hat Y; \hat Z, \hat C| \hat W)$,
which is a straightforward generalisation of the single-channel figure of merit
$I(Y;M|W)$.
Here the hat indicates a vector in which each entry comes from one TV channel;
the $Y$ is the colluders' output, the $M$ stands for the coalition's symbol tally in the single TV channel case, while $(\hat Z, \hat C)$ represent the coalition's symbol tally in the multiple TV channel case, and
$W$ is the bias. Although the generalised figure of merit looks simple, the multiple TV channel maximin game is more complicated
than the single TV channel case. If a pirate is active in one TV channel, then this excludes the possibility that they are active in another TV channel. This exclusion causes a nontrivial dependence between the TV channels, which complicates the analysis:
it is not a priori clear if the multi-channel attack can be treated as a set of independent single-channel attacks.

\begin{itemize}[leftmargin=5mm,itemsep=0mm]
\item
We find the solution of the maximin game for the $k \to \infty$ limit of the
payoff function $I(\hat Y; \hat Z, \hat C| \hat W)$,
using the same technique as Huang and Moulin \cite{HM2012TIFS}.
The optimal bias distribution for the tracer is the arcsine distribution, and the
optimal colluder strategy within each channel is Interleaving.
Moreover, it is optimal for the attackers to spread evenly over the channels.
Although the result is far from surprising,
the proof is less simple than one would have hoped for.
The proof needs some careful handling of expressions with different orders in~$1/k$ that arise
from different attack strategies for spreading out over the channels.
\item
We present an alternative payoff functional, namely the mutual information $I(\hat Y;\hat X_{\cal K}|\hat W)$.
Here $\hat X_{\cal K}$ stands for the part of the code matrices in all the TV channels that can potentially be
tuned into by the coalition~$\cal K$.
We argue that the two payoffs have the same maximin game asymptotically, and this leads us to conjecture that the optimal strategies hold for $I(\hat Y;\hat X_{\cal K}|\hat W)$ as well,
and that the fingerprinting capacity asymptotically behaves like $L^2/(k^2 2\ln2)$.
\end{itemize}

In Section~\ref{sec:model} we introduce the multiple TV channels model and the payoff function.
In Section~\ref{sec:theorem} we derive the maximin solution for the asymptotic payoff.
We discuss the alternative payoff and fingerprinting capacity in Section~\ref{sec:alt}.
We summarize and suggest future work in Section~\ref{sec:discussion}. Before we proceed, we introduce common notations used throughout the manuscript.
\subsection{Notation}
Let $m$ be the number of users; $\mathcal{M}\overset{\Delta}{=}\{1,2,\dots,m\}$ the index set for all users; $\mathcal{X}\overset{\Delta}{=}\{0,1,\dots,q-1\}$ denote the q-ary fingerprinting alphabet; $n$ denote the code length; $\mathcal{K}\overset{\Delta}{=} \{j_1,\dots,j_K\} \subset \mathcal{M}$ the index set of the coalition, where $K$ is the number of colluders; $k$ is the nominal coalition size\footnote{In this work, we do not make a distinction between $K$ and $k$, i.e. we assume the real number of colluders is known to the tracer.}; $L$ the number of TV channels; $\mathcal{L}\overset{\Delta}{=}\{1,2,\dots,L\}$ the index set for all TV channels; $\tilde{L}$ the number of TV channels a user can tune in to simultaneously; $(\cdot)_{\mathcal{K}}\overset{\Delta}{=} \{(\cdot)_j :j \in \mathcal{K}\}$; vectors are denoted by boldface letters; $L-$tuples are denoted by $\hat{(\cdot)}\overset{\Delta}{=} ((\cdot)^1,\dots,(\cdot)^L)$; $|(\cdot)|$ denotes the $L^1$-norm of a vector or the cardinality of a set, depending on the argument. We denote the Kronecker delta by $\delta(i,j)$, which is equal to 1 when $i=j$, and 0 otherwise. We use the following notation for asymptotic relations: Let $f(k)$ and $g(k)$ be two functions defined on the real numbers. $f(k)=O(g(k))$ if $\exists c >0, k^{*}>0$ such that $f(k) \leq c g(k)$, $\forall k \geq k^{*}$. $f(k)=o(g(k))$ if $f(k)/g(k)$ tends to 0. $f(k)=w(g(k))$ if $f(k)/g(k)$ tends to $\infty$. $f(k) \sim g(k)$ if $f(k)/g(k)$ tends to a non-zero constant. $f(k) \rightarrow g(k)$ if $f(k)/g(k)$ tends to 1.
\section{Channel Model}
\label{sec:model}

\subsection{Channel law}
In the single channel case \cite{HM2012TIFS}, the tracer produces codewords for $m$ users in a random fashion. This is achieved as follows: for each of the $n$ segments a bias vector $\mathbf{W}$ is drawn from a distribution $f_{\mathbf{W}}$ chosen by the tracer, and $X_j \in \{0,1,\dots,q-1\}$ is assigned for user $j$ according to a categorical distribution with parameters $\{W_i\}_{i=0}^{q-1}$:
\begin{equation}
\mathbb{P}(X_j=x|\mathbf{W}=\mathbf{w})=w_x.
\end{equation}
The users are assigned their codewords independently by the tracer, and hence we also have:
\begin{equation}
\mathbb{P}(X_{1}=x_{1},\dots,X_{m}=x_{m}|\mathbf{W}=\mathbf{w})=\prod_{j=1}^{m} w_{x_{j}}.
\end{equation}
Note that we do not include a segment index since this procedure is repeated for all segments independently, hence producing the codewords $\{\mathbf{X}_j\}_{j=1}^{m}$, where $\mathbf{X}_j=(X_{j, 1},\dots, X_{j, n})$.

We now describe an analogous procedure for producing the codewords of $m$ users in the case of multiple TV channels. We adopt the same notation as in the case of a single TV channel, and one can recover the single TV channel description by setting $L=1$ in what follows. The tracer produces codewords $\{\hat{\mathbf{X}}_1,\dots,\hat{\mathbf{X}}_m\}$ with $\mathbf{X}_j^l \in \{0,1,\dots,q-1\}^n$ denoting the $l$-th codeword for user $j$. This is done by choosing $L$ i.i.d bias vectors at each segment, and each segment is independent of other segments. That is, let $\hat{\mathbf{W}}_i$ be the bias vectors at segment $i$, then:
\begin{equation}
\begin{aligned}
f_{\hat{\mathbf{W}}_1,\dots,\hat{\mathbf{W}}_n}(\hat{\mathbf{w}}_1,\dots,\hat{\mathbf{w}}_n)&=\prod_{i=1}^nf_{\hat{\mathbf{W}}_i}(\hat{\mathbf{w}}_i)=\prod_{i=1}^n\prod_{l=1}^Lf_{\mathbf{W}^l_i}(\mathbf{w}^l_i)=\prod_{i=1}^n\prod_{l=1}^Lf_{\mathbf{W}}(\mathbf{w}^l_i),
\end{aligned}
\end{equation}
For any $l \in \{1,\dots,L\}$ and $i \in \{1,\dots,n\}$, $\{X^l_{j,i}\}_{j=1}^{m}$ are i.i.d and drawn from a categorical distribution:
\begin{equation}
\mathbb{P}(X^l_{1,i}=x_1,\dots,X^l_{m,i}=x_m|\mathbf{W}^l_{i}=\mathbf{w})=\prod_{j=1}^{m} w_{x_j}.
\end{equation}
The pirates receive codewords $\hat{\mathbf{X}}_{\mathcal{K}}$, and produce the output $\hat{\mathbf{Y}}$, where $\mathbf{Y}^l \in \mathcal{Y}^n$ (we assume an RDM setting, i.e. $\mathcal{Y}=\mathcal{X}$), according to a pmf $p_{\hat{\mathbf{Y}}|\hat{\mathbf{X}}_{\mathcal{K}},\hat{\mathbf{S}}}$, where $\{\hat{S}_i\}_{i=1}^n$ is the assignment of pirates to channels at each segment, i.e. $S^l_i$ is a random subset of $\mathcal{K}$ that we assume is independent of $X^l_{\mathcal{K}, i}$. Let $\tilde{L}$ be the maximum number of channels a single user can simultaneously tune in to\footnote{We assume pirates have the same accessability constraints as normal users, i.e. access to the same hardware.}, then for any segment $i$ and TV channel $l$, a realization $s^l_i$ is an assignment
\begin{equation}
s^l_i: l \mapsto 2^{\mathcal{K}},
\end{equation}
that respects the following:
\begin{align}\label{s}
  &s^l_i\neq \emptyset, \\
  &\bigcup_{l=1}^{L}s^l_i=\mathcal{K},\\
  \label{nomore}
  &\bigcap_{l\in \tilde{\mathcal{L}}}s_i^{l}=\emptyset, \hspace{1mm} \forall \tilde{\mathcal{L}} \subseteq \mathcal{L}, \ |\tilde{\mathcal{L}}| >\tilde{L}.
\end{align}
Equation \eqref{nomore} states that a single pirate can not be assigned to more than $\tilde{L}$ TV channels.
We further assume memoryless-ness and no feedback of the collusion channel, which implies:
\begin{equation}
\label{claw1}
p_{\hat{\mathbf{Y}}|\hat{\mathbf{X}}_{\mathcal{K}},\hat{\mathbf{S}}}(\hat{\mathbf{y}}|\hat{\mathbf{x}}_{\mathcal{K}},\hat{\mathbf{s}})=\prod_{i=1}^{n}p_{\hat{Y}|\hat{X}_{\mathcal{K}},\hat{S}}(\hat{y}_i|\hat{x}_{\mathcal{K},i},\hat{s}_i),
\end{equation}
and in what follows, we drop the segment index on random variables. In this work, we restrict to the case of each user having access to only one of the TV channels at each segment, i.e. $\tilde{L}=1$ (and so $K\geq L$). In this case $\{s^l\}_{l=1}^L$ form a partition of $\mathcal{K}$, or more precisely, $\hat{s}$ defines a weak ordering of $\mathcal{K}$ where pirates assigned to the same TV channel are tied. The size of the support for $\hat{S}$ (denoted by $\hat{\mathcal{S}}$) is then given by\footnote{$\stirling{K}{L}$ counts the number of ways we can partition $\mathcal{K}$ into $L$ non-empty subsets, while the $L!$ factor orders them across TV channels. Otherwise, it can be viewed as putting $K$ distinguishable balls into $L$ distinguishable bins.}:
\begin{align}
|\hat{\mathcal{S}}|=L!\,\stirling{K}{L},
\end{align}
where $\stirling{a}{b}$ denotes Stirling numbers of the second kind. If we instead disregard the identity of the colluders assigned to each channel, and only keep their numbers, the support size would shrink to\footnote{It can be viewed as putting $K$ indistinguishable balls into $L$ distinguishable bins, which can be easily proved using stars and bars.} $\binom{K-1}{L-1}$. Furthermore, if we also disregard the TV channel labels, the support size would become $P(K,L)$, which is the number of partitions of $K$ into exactly $L$ parts. It is common in the literature to assume colluder symmetry, i.e. all pirates share the risk equally.  Furthermore, it is logical to assume TV channel symmetry, i.e. all $L$ TV channels are equally important. Therefore in practice, finding an optimal distribution over $\{\hat{s}\}$ is the same as finding an optimal distribution over $P(K,L)$ partitions. Since $\tilde{L}=1$, i.e. $\{s^l\}_{l=1}^{L}$ are disjoint, then we can simplify equation \eqref{claw1} by writing:
\begin{equation}
\label{claw2}
p_{\hat{\mathbf{Y}}|\hat{\mathbf{X}}_{\mathcal{K}},\hat{\mathbf{S}}}(\hat{\mathbf{y}}|\hat{\mathbf{x}}_{\mathcal{K}},\hat{\mathbf{s}})=\prod_{i=1}^{n}\prod_{l=1}^{L}p_{Y^l|X^l_{S^l},S^l}(y^l_i|x^l_{s^l_i,i},s^l_i),
\end{equation}
i.e. colluders on different TV channels do not communicate after being assigned. Applying TV channel symmetry to \eqref{claw2}, i.e. removing any bias towards a particular TV channel in the pirates' strategy (it is true for the tracer, i.e. $\{X^l_{\mathcal{K}}\}_{l=1}^{L}$ are i.i.d.), then we can write:
\begin{equation}
\label{claw3}
p_{\hat{\mathbf{Y}}|\hat{\mathbf{X}}_{\mathcal{K}},\hat{\mathbf{S}}}(\hat{\mathbf{y}}|\hat{\mathbf{x}}_{\mathcal{K}},\hat{\mathbf{s}})=\prod_{i=1}^{n}\prod_{l=1}^{L}p_{Y|X_{S},S}(y^l_i|x^l_{s^l_i,i},s^l_i).
\end{equation}
Given $x_s$ (a realization of $X_S$), we define the tally vector $\mathbf{m}$ as the following:
\begin{align}
\mathbf{m} &\overset{\Delta}{=} (m_0, \dots, m_{q-1}),
\\
m_{\alpha}&\overset{\Delta}{=} |\{j \in s : x_j=\alpha\}|.
\end{align}
Imposing colluder symmetry, we can write:
\begin{equation}
\label{claw4}
p_{\hat{\mathbf{Y}}|\hat{\mathbf{X}}_{\mathcal{K}},\hat{\mathbf{S}}}(\hat{\mathbf{y}}|\hat{\mathbf{x}}_{\mathcal{K}},\hat{\mathbf{s}})=\prod_{i=1}^{n}\prod_{l=1}^{L}p_{Y|\mathbf{M}}(y^l_i|\mathbf{m}^l_i),
\end{equation}
where $\mathbf{m}^l_i$ is the tally vector received by the pirates assigned to the $l$-th TV channel at segment $i$, and we assume $p_{Y|\mathbf{M}}$ abides by the marking assumption. Note that unlike the single TV channel\footnote{In the single channel case this would be equal to $K$, which is usually unknown but constant.} case ($L=1$), $|\mathbf{M}^l|=|S^l| \overset{\Delta}{=} C^l$ is a random variable, which is the number of colluders assigned to attack TV channel $l$. Another implication of the TV channel symmetry condition is the following:
\begin{equation}
\label{csymm}
p_{C^l}(c^l)=p_{C}(c^l),
\end{equation}
i.e. the number of pirates assigned to attack different TV channels must be identically distributed. This follows since $p_{\hat{C}}$ together with $p_{Y|\mathbf{M}}$ form the pirates' strategy. Figure \ref{fig} showcases the process of producing the colluders outputs across the $L$ TV channels. Note that from colluder symmetry, $p_{\hat{S}|\hat{C}}$ is specified by the rule:
\begin{equation}
p_{\hat{S}|\hat{C}}(\hat{s}|\hat{c})=\begin{cases}
\dfrac{1}{\binom{k}{c^1, c^2, \dots, c^L}},& |s^l|=c^l, \ \forall l \in \mathcal{L}\\
0,& \mathrm{otherwise}
\end{cases},
\end{equation}
where $\binom{k}{c^1, c^2, \dots, c^L}$ is the multinomial coefficient.
\begin{lemma}
\label{mlemma}
  $\{\mathbf{M}^l\}_{l=1}^L$ are identically distributed.
\end{lemma}
\begin{proof}\footnote{Summations over $\mathbf{w}$ can be appropriately replaced by integrals in the case of $f_{\mathbf{W}}$ being a density.}
\begin{equation}
\begin{aligned}
&p_{\mathbf{M}^l}(\mathbf{m})=\sum_{\mathbf{w}}p_{\mathbf{M}^l|\mathbf{W}^l}(\mathbf{m}|\mathbf{w})f_{\mathbf{W}}(\mathbf{w})
\\
&=\sum_{c}\sum_{\mathbf{w}}p_{C^l}(c)p_{\mathbf{M}^l|\mathbf{W}^l, C^l}(\mathbf{m}|\mathbf{w}, c)f_{\mathbf{W}}(\mathbf{w})
\\
&=\sum_{c}\sum_{\mathbf{w}}p_{C}(c)p_{\mathbf{M}^l|\mathbf{W}^l, C^l}(\mathbf{m}|\mathbf{w}, c)f_{\mathbf{W}}(\mathbf{w})
\\
&\overset{\mathrm{(a)}}{=}\sum_{c}\sum_{\mathbf{w}}p_{C}(c)p_{\mathbf{M}|\mathbf{W}, C}(\mathbf{m}|\mathbf{w}, c)f_{\mathbf{W}}(\mathbf{w})=p_{\mathbf{M}}(\mathbf{m}),
\end{aligned}
\end{equation}
where in $\mathrm{(a)}$ we used the fact that $p_{\mathbf{M}^l|\mathbf{W}^l, C^l}$ is a multinomial distribution regardless of $l$.
\end{proof}
\begin{lemma}
\label{ylemma}
$\{Y^l\}_{l=1}^L$ are identically distributed.
\end{lemma}
\begin{proof}
\begin{equation}
\begin{aligned}
p_{Y^l}(y^l)=\sum_{\mathbf{m}^l}p_{Y^l|\mathbf{M}^l}(y^l|\mathbf{m}^l)p_{\mathbf{M}^l}(\mathbf{m}^l)=\sum_{\mathbf{m}^l}p_{Y|\mathbf{M}}(y^l|\mathbf{m}^l)p_{\mathbf{M}}(\mathbf{m}^l)=p_Y(y^l),
\end{aligned}
\end{equation}
where the equalities follow from {Lemma \ref{mlemma}} and TV channel symmetry ($p_{Y|\mathbf{M}}$ is part of the pirates' strategy).
\end{proof}
\begin{figure}[!ht]
\centerline{\includegraphics[scale=.4, trim={0cm 4cm 2cm 0cm}]{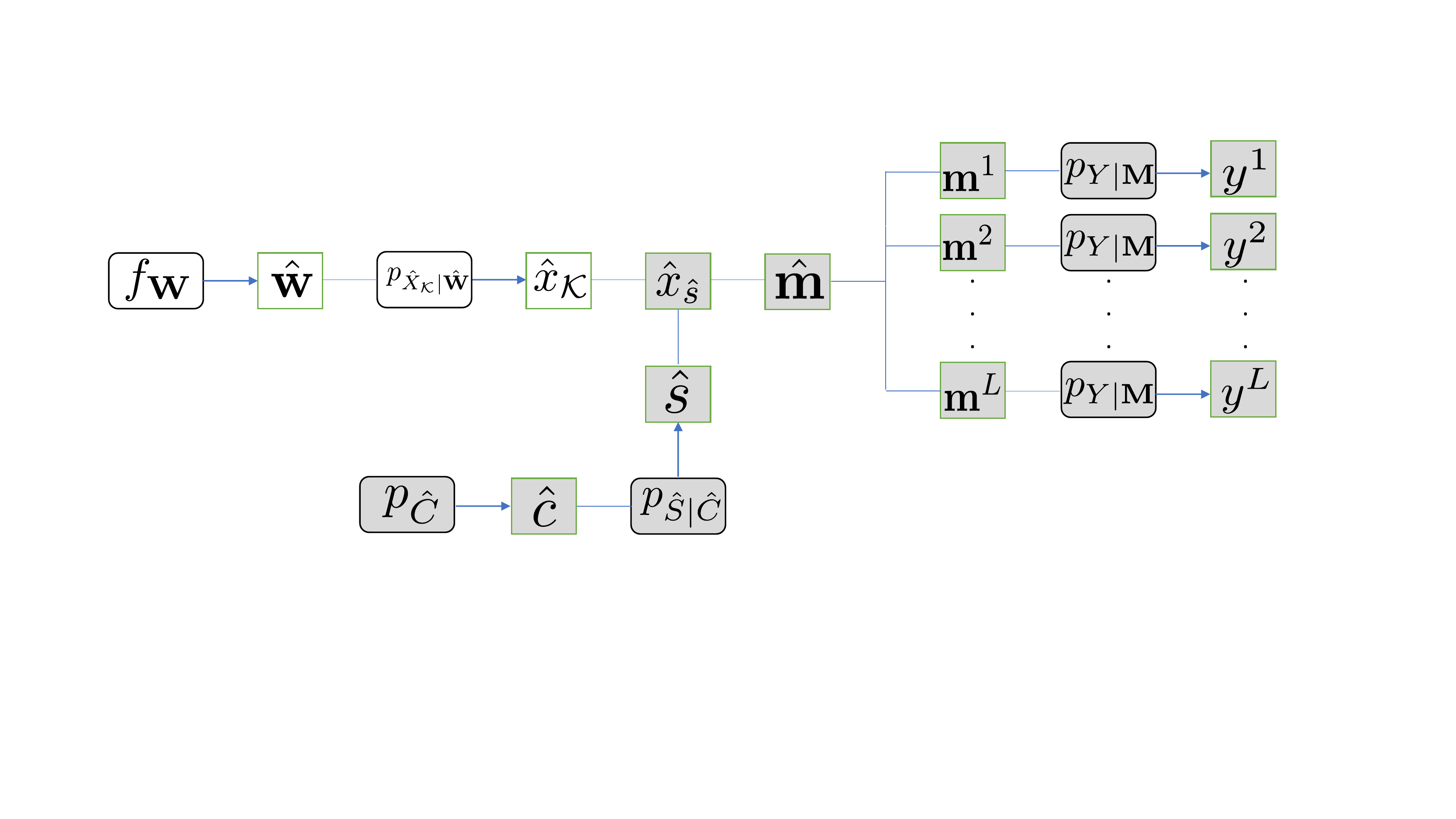}}
\caption{\textit{A schematic diagram that shows the watermarking/collusion process. Arrows indicate drawing realizations from a pmf, while straight lines indicate use of those realizations. White boxes correspond to the tracer, while grey boxes correspond to the coalition.}}
\label{fig}
\end{figure}

\subsection{Payoff functional}
For any $l \in \{1,\dots,L\}$, we define the single channel payoff as:
\begin{equation}
\label{scpayoff}
\dfrac{I(Y^l;\mathbf{M}^l|C^l=c, \mathbf{W}^l)}{c}=\dfrac{I(Y;\mathbf{M}|C=c, \mathbf{W})}{c}
\overset{\Delta}{=}
\mathcal{I}_{c}(\mathbf{W}, p_{Y|\mathbf{M}, C=c}),
\end{equation}
where equality follows from {Lemma \ref{mlemma}}, {Lemma \ref{ylemma}}, and $\{\mathbf{W}^l\}_{l=1}^L$ being i.i.d. In analogy to the single channel case \cite{HM2012TIFS}, we define the multi channel payoff as the following:
\begin{equation}
\label{multipayoff}
J_{k, L}[f_{\mathbf{W}}, p_{Y|\mathbf{M}}, p_{\hat{C}}]\overset{\Delta}{=} \dfrac{I(\hat{Y};\hat{\mathbf{M}}|\hat{\mathbf{W}})}{k}.
\end{equation}
From here on, we restrict our attention to the binary alphabet case $\mathcal{X}=\{0,1\}=\mathcal{Y}$, and we wish to evaluate $J_{k, L}[f_{\mathbf{W}}, p_{Y|\mathbf{M}}, p_{\hat{C}}]$ for this case. We define the random variables $\{Z^l\}_{l=1}^{L}$ by:
\begin{align}
Z^l\overset{\Delta}{=}\sum_{j \in S^l}X^l_j,
\end{align}
which counts the number of 1s received by the pirates assigned to the $l$-th TV channel. Equation \eqref{csymm} and {Lemma \ref{mlemma}} imply that $\{Z^l\}_{l=1}^{L}$ and $\{C^l\}_{l=1}^{L}$ are sets of identically distributed random variables. Therefore, we can talk about $Z$ and $C$ with no reference to a TV channel label. We also use the scalar bias $W$ instead of $\mathbf{W}$, which is the probability of assigning a user the symbol 1. Using this notation, we define the following:
\begin{align}
&p_{Z|C,W}(z|c,w)\overset{\Delta}{=} \alpha_z(w,c)=\begin{cases}\binom{c}{z}w^z(1-w)^{c-z},& 0\leq z \leq c \\0,& \text{otherwise} \end{cases},\\
&p_{Z,C|W}(z,c|w)=p_{C|W}(c|w)p_{Z|C,W}(z|c,w)=p_C(c)\alpha_z(c,w),\\
&\pi_{z,c}\overset{\Delta}{=} p_{Y|Z,C}(1|z,c), \quad 0\leq z \leq c.
\end{align}
Note that $(Z^l, C^l)$ is a sufficient statistic for determining $Y^l$, while $(W^l, C^l)$ is a sufficient statistic for determining $Z^l$, therefore we have the following (keeping in mind that $\{W^l\}_{l=1}^L$ are i.i.d.):
\begin{equation}
\label{someeq1}
\begin{aligned}
p_{\hat{Y}|\hat{Z}, \hat{C}, \hat{W}}(\hat{y}|\hat{z}, \hat{c}, \hat{w})&=\prod_{l=1}^L p_{Y^l|Z^l, C^l}(y^l|z^l, c^l)\\
&=\prod_{l=1}^L p_{Y^l|Z^l, C^l, W^l}(y^l|z^l, c^l, w^l)\\
&=\prod_{l=1}^Lp_{Y|Z, C, W}(y^l|z^l, c^l, w^l),
\end{aligned}
\end{equation}
\begin{equation}
\label{someeq2}
p_{\hat{Z}|\hat{C}, \hat{W}}(\hat{z}|\hat{c}, \hat{w})=\prod_{l=1}^L p_{Z^l|C^l, W^l}(z^l| c^l, w^l)
=\prod_{l=1}^L\alpha_{z^l}(w^l, c^l).
\end{equation}
\begin{lemma}
\label{somelemma}
$p_{\hat{Y}|\hat{C}, \hat{W}}(\hat{y}|\hat{c}, \hat{w})=\prod_{l=1}^{L}p_{Y|C, W}(y^l|c^l, w^l)$.
\end{lemma}
\begin{proof}
\begin{equation}
\label{someeq3}
\begin{aligned}
&p_{\hat{Y}|\hat{C}, \hat{W}}(\hat{y}|\hat{c}, \hat{w})=\sum_{\hat{z}}p_{\hat{Y}|\hat{C}, \hat{Z}, \hat{W}}(\hat{y}|\hat{c}, \hat{z}, \hat{w})p_{\hat{Z}|\hat{C}, \hat{W}}(\hat{z}|\hat{c}, \hat{w})\\
&=\sum_{z^1}\cdots \sum_{z^L}\prod_{l=1}^L p_{Y^l|Z^l, C^l, W^l}(y^l|z^l, c^l, w^l)p_{Z^l|C^l, W^l}(z^l| c^l, w^l)\\
&=\sum_{z^1}\cdots \sum_{z^L}\prod_{l=1}^L p_{Y, Z|C, W}(y^l, z^l|c^l, w^l)=\prod_{l=1}^L p_{Y|C, W}(y^l|c^l, w^l).
\end{aligned}
\end{equation}
\end{proof}
\begin{lemma}
\label{scpaylemma}
$I(\hat{Y}; \hat{Z}|\hat{C}=\hat{c}, \hat{W})=\sum_{l=1}^L c^l\mathcal{I}_{c^l}(W, \bm{\pi}_{c^l})$, $\bm{\pi}_c\overset{\Delta}{=}(\pi_{0,c},\cdots,\pi_{c,c})$.
\end{lemma}
\begin{proof}
\begin{equation}
\begin{aligned}
&I(\hat{Y};\hat{Z}|\hat{C}=\hat{c}, \hat{W}=\hat{w})\overset{\mathrm{(a)}}{=}H(\hat{Y}|\hat{C}=\hat{c}, \hat{W}=\hat{w})-H(\hat{Y}|\hat{Z}, \hat{C}=\hat{c}, \hat{W}=\hat{w})&\\
&\overset{\mathrm{(b)}}{=}H(\hat{Y}|\hat{C}=\hat{c}, \hat{W}=\hat{w})-\sum_{\hat{z}}p_{\hat{Z}|\hat{C}, \hat{W}}(\hat{z}|\hat{c}, \hat{w}) H(\hat{Y}|\hat{Z}=\hat{z}, \hat{C}=\hat{c}, \hat{W}=\hat{w})\\
&\overset{\mathrm{(c)}}{=}\sum_{l=1}^LH(Y|C=c^l, W=w^l)
\\
&-\sum_{\hat{z}}\left(\prod_{l'=1}^L p_{Z|C, W}(z^{l'}| c^{l'}, w^{l'})\right)\sum_{l=1}^L H(Y|Z=z^l, C=c^l, W=w^l)
\\
&\overset{\mathrm{(d)}}{=}\sum_{l=1}^LH(Y|C=c^l, W=w^l)
-\sum_{l=1}^L\sum_{z^l}p_{Z|C, W}(z^{l}| c^{l}, w^{l})H(Y|Z=z^l, C=c^l, W=w^l)\\
&\overset{\mathrm{(e)}}{=}\sum_{l=1}^L I(Y;Z|C=c^l, W=w^l),
\end{aligned}
\end{equation}
where $\mathrm{(a)}$ and $\mathrm{(b)}$ follow from the definitions of conditional mutual information and conditional entropy; $\mathrm{(c)}$ follows from equations \eqref{someeq1}, \eqref{someeq2}, and {Lemma \ref{somelemma}}; $\mathrm{(d)}$ follows by direct computation; $\mathrm{(e)}$ follows from the definition of conditional mutual information. Finally, it is straightforward to confirm that
\begin{equation}
I(\hat{Y};\hat{Z}|\hat{C}=\hat{c}, \hat{W})=\sum_{l=1}^L I(Y;Z|C=c^l, W).
\end{equation}
\end{proof}
\begin{lemma}
$\mathcal{I}_c(W, \bm{\pi}_{c})$ coincides with the single channel payoff for $c$ pirates as defined in \cite{HM2012TIFS}.
\end{lemma}
\begin{proof}
The pirate strategy consists of $p_{\hat{C}}$, as well as the vectors $\{\bm{\pi}_{c'}\}_{c'=1}^{k-L+1}$. Note that $\mathcal{I}_c(W, \bm{\pi}_{c})$ is independent of $p_{\hat{C}}$ and all $\{\bm{\pi}_{c'}: c' \neq c\}$, i.e. it only depends on $f_W$ and $\bm{\pi}_{c}$.
\end{proof}
\begin{lemma}
\label{ineqlemma}
$2 \ln 2 \hspace{0.5mm}c^2 \mathcal{I}_c(W, \bm{\pi}_{c}) \geq 4\left[\mathlarger{\int}_0^1 \dfrac{dw}{f_W(w)w(1-w)}\right]^{-1}$.
\end{lemma}
\begin{proof}
This follows from \textit{Theorem 7} in \cite{HM2012TIFS}.
\end{proof}

\section{Asymptotic Theorem}
\label{sec:theorem}
We wish to obtain the asymptotic behavior of the payoff given in equation \eqref{multipayoff} (for $q=2, \tilde{L}=1$), as well as the solution to the asymptotic maximin game in analogy to the single TV channel case. For the single TV channel case \cite{HM2012TIFS}, this is referred to as the asymptotic saddle point value, which is the limit of the maximin value of the payoff as the number of pirates is sent to infinity. This is also accompanied by the asymptotic optimal strategies for both tracer and pirates, which are the asymptotic solutions to the maximin game. However, as this makes logical sense in terms of taking the limit of the maximin game, this is usually not solved in the same way as it is formulated. Instead of taking the limit of a sequence of optimal payoffs (and optimal strategies), which is not possible given no closed form solution exists for either, the payoff is approximated for large number of pirates and solved in the limit. This can be justified if the optimal payoff converges uniformly in the limit, which is usually glossed over in the literature. Nevertheless, we shall adopt the same approach and write down an expansion of our payoff function. In analogy to the single channel case, we make the following regularity assumptions:
\begin{itemize}
  \item For any $c \in \{1,\dots, k-L+1\}$, there exists a bounded, twice differentiable function $g_c(x)$ for $x \in [0,1]$ with $g(0)=0$ and $g(1)=1$ such that:
\begin{align}
&\pi_{z,c}=g_c\left(\dfrac{z}{c}\right), \quad \forall z \in \{0, \dots, c\},
&\lim_{k \to \infty} g_c(x) \overset{\Delta}{=} g(x),\hspace{2mm} c \in w(1).
\label{klimitg}
\end{align}
\end{itemize}
\begin{lemma}
\label{asymhm}
Let $G[g_c(w)]\overset{\Delta}{=} \arccos\left[1-2g_c(w)\right]$, and $\mathcal{J}[g_c(w)] \overset{\Delta}{=} w(1-w) \left[G'(w)\right]^2$, where the derivative is w.r.t $w$. Let $\bar{\mathcal{J}}[g_c(W)]\overset{\Delta}{=} \mathbb{E}_{f_W}\left[\mathcal{J}\left[g_c(W)\right]\right]$ then we can write:
\begin{equation}
\mathcal{I}_{c}(W,\bm{\pi}_{c})=\dfrac{1}{c^22 \ln 2 }\bar{\mathcal{J}}[g_c(W)]+o\left(\dfrac{1}{c^2}\right).
\end{equation}
\end{lemma}
\begin{proof}
This follows from \textit{Theorem 8} in \cite{HM2012TIFS}.
\end{proof}
\begin{lemma}
\label{asymhm1}
$\bar{\mathcal{J}}[g(W)] \geq \pi^2\left[\mathlarger{\int}_0^1 \dfrac{dw}{f_W(w)w(1-w)}\right]^{-1}$, with equality iff $g(w)=g_{\mathrm{opt}}(w)$ given by:
\begin{equation}
\label{gopt}
g_{\mathrm{opt}}(w)=\dfrac{1}{2}\left[1-\cos \left(\dfrac{\pi \mathlarger{\int}_0^w \dfrac{dv}{f_W(v)v(1-v)}}{\mathlarger{\int}_0^1 \dfrac{dv}{f_W(v)v(1-v)}}\right)\right].
\end{equation}
\end{lemma}
\begin{proof}
This is \textit{Lemma 7} in \cite{HM2012TIFS}.
\end{proof}
For $q=2, \tilde{L}=1$, we can rewrite equation \eqref{multipayoff} as:
\begin{equation}
\label{payoff}
\begin{aligned}
J_{k,L}[f_W,p_{Y|Z, C},p_{\hat{C}}]&=k^{-1}I(\hat{Y}; \hat{Z}|\hat{C}, \hat{W})+k^{-1}I(\hat{Y}; \hat{C}|\hat{W})
\\
&=k^{-1}\sum_{\hat{c}}p_{\hat{C}}(\hat{c})\sum_{l=1}^L c^l \mathcal{I}_{c^l}(W,\bm{\pi}_{c^l})+k^{-1}I(\hat{Y}; \hat{C}|\hat{W})
\\
&=k^{-1}L\sum_{c}p_{C}(c)c \mathcal{I}_{c}(W,\bm{\pi}_{c})+k^{-1}I(\hat{Y}; \hat{C}|\hat{W}),
\end{aligned}
\end{equation}
where the second equality follows from {Lemma \ref{scpaylemma}}, and the third equality follows from $\{C^l\}_{l=1}^L$ having the same marginal, as dictated by $p_{\hat{C}}$ being symmetric under permutations (TV channel symmetry), i.e. $C$ denotes the number of pirates assigned to any TV channel.
\begin{lemma}
\label{jstar}
Let $\tilde{J}_{k, L}[f_W, p_{Y|Z, C}, p_C]\overset{\Delta}{=} \sum_{c}p_{C}(c)\dfrac{c}{k} \mathcal{I}_{c}(W,\bm{\pi}_{c})$, then an asymptotically optimal\footnote{An asymptotically optimal sequence of strategies for the pirates is one that produces the lowest value for the payoff in the limit.} $p_C$ for $\tilde{J}$ must obey:
\begin{equation}
\lim_{k \to \infty}\dfrac{p_C(c)k^\alpha}{c^\alpha}< \infty,
\end{equation}
for some $\alpha \geq1$, and $\tilde{J}$ optimally decays like $1/k^2$.
\end{lemma}
\begin{proof}
We know that $c$ can at most grow as a fraction of $k$, that is, the value $q \overset{\Delta}{=} c/k$ can only belong to $[0,1]$ in the limit. That is, we can divide the summation in $\tilde{J}$ as follows:
\begin{equation}
\label{ft2}
\begin{aligned}
\tilde{J}=&\sum_{c \sim k}p_{C}(c)\dfrac{c}{k} \left(\dfrac{1}{c^22 \ln 2 }\bar{\mathcal{J}}[g_c(W)]+o\left(\dfrac{1}{c^2}\right)\right)\\
&+\sum_{c \in w(1)\cap o(k)}p_{C}(c)\dfrac{c}{k} \left(\dfrac{1}{c^22 \ln 2 }\bar{\mathcal{J}}[g_c(W)]+o\left(\dfrac{1}{c^2}\right)\right)\\
&+\sum_{c \in O(1)}p_{C}(c)\dfrac{c}{k} \mathcal{I}_{c}(W,\bm{\pi}_{c}),
\end{aligned}
\end{equation}
where note that $\bar{\mathcal{J}}$ is bounded in the limit. For non-decaying measure $p_C(c)$, we can see from \eqref{ft2} that the first term decays like $1/k^2$, the second term decays like $1/kc$ where $c$ is sub-linear in $k$, and the last term decays like $1/k$. Therefore if $p_C(c)$ in the last two terms is not small enough in the limit, the payoff will decay slower than $1/k^2$. Any strategy that assigns asymptotically small enough measure in the last two terms will hence produce a lower value for the payoff given large enough $k$, nevertheless, the first term will always decay like $1/k^2$.
\end{proof}
\begin{lemma}
$\mathbb{E}_{p_C}\left[\dfrac{C}{k}\right]=\dfrac{1}{L}$.
\end{lemma}
\begin{proof}
\begin{equation}
\begin{aligned}
&\sum_{c}p_C(c)\dfrac{c}{k}=\dfrac{1}{L}\sum_{l=1}^{L}\sum_{c}p_C(c)\dfrac{c}{k}=\dfrac{1}{L}\sum_{l=1}^{L}\sum_{c^l}p_C(c^l)\dfrac{c^l}{k}\\
&=\dfrac{1}{L}\sum_{c^1,\dots,c^L}p_{C^1,\dots,C^L}(c^1,\dots,c^L)\dfrac{c^1+\dots+c^L}{k}=\dfrac{1}{L}.
\end{aligned}
\end{equation}
\end{proof}
\begin{lemma}
\label{convjen}
$\sum_c p_C(c)\dfrac{k}{c} \geq L$, with equality iff $p_C(k/L)=1$.
\end{lemma}
\begin{proof}
\begin{equation}
  \mathbb{E}_{p_C}\left[\dfrac{k}{C}\right]=k\mathbb{E}_{p_C}\left[\dfrac{1}{C}\right]\geq k\dfrac{1}{\mathbb{E}_{p_C}\left[C\right]}=\dfrac{1}{\mathbb{E}_{p_C}\left[\dfrac{C}{k}\right]}=L,
\end{equation}
which is an application of Jensen's inequality to a strictly convex function.
\end{proof}
\begin{lemma}
\label{lemmaL}
Let $\alpha \geq 1$, then $\sum_{c\sim k}p_C(c)\dfrac{k}{c}\geq L+\cdots$, where terms in ellipsis go to zero in the limit $k \to \infty$.
\end{lemma}
\begin{proof}
\begin{equation}
\begin{aligned}
\left[\sum_{c'\sim k}p_C(c')\right]^{-1}\sum_{c\sim k}p_C(c)\dfrac{c}{k}=\dfrac{1}{L}\left[\sum_{c'\sim k}p_C(c')\right]^{-1}-\left[\sum_{c'\sim k}p_C(c')\right]^{-1}\sum_{c \nsim k}p_{C}(c)\dfrac{c}{k},
\end{aligned}
\end{equation}
and so applying Jensen's inequality, we get:
\begin{equation}
\begin{aligned}
\left[\sum_{c'\sim k}p_C(c')\right]^{-1}\sum_{c\sim k}p_C(c)\dfrac{k}{c}\geq\left[\dfrac{1}{L}\left[\sum_{c'\sim k}p_C(c')\right]^{-1}-\left[\sum_{c'\sim k}p_C(c')\right]^{-1}\sum_{c \nsim k}p_{C}(c)\dfrac{c}{k}\right]^{-1},
\end{aligned}
\end{equation}
and if we multiply both sides by $\sum_{c'\sim k}p_C(c')$, we arrive at:
\begin{equation}
\sum_{c\sim k}p_C(c)\dfrac{k}{c}\geq\left[\dfrac{1}{L}\left[\sum_{c'\sim k}p_C(c')\right]^{-2}-\left[\sum_{c'\sim k}p_C(c')\right]^{-2}\sum_{c \nsim k}p_{C}(c)\dfrac{c}{k}\right]^{-1},
\end{equation}
or simply:
\begin{equation}
\begin{aligned}
&\sum_{c\sim k}p_C(c)\dfrac{k}{c}\geq L\dfrac{\left[\sum_{c'\sim k}p_C(c')\right]^{2}}{1-L\sum_{c \nsim k}p_{C}(c)\dfrac{c}{k}}=L\dfrac{\left[1-\sum_{c'\nsim k}p_C(c')\right]^{2}}{1-L\sum_{c \nsim k}p_{C}(c)\dfrac{c}{k}}
\\
&=L\dfrac{\left[1-\sum_{c'\nsim k}\bar{p}_C(c')\left(\dfrac{c}{k}\right)^{\alpha}\right]^{2}}{1-L\sum_{c \nsim k}\bar{p}_{C}(c)\left(\dfrac{c}{k}\right)^{\alpha+1}},
\end{aligned}
\end{equation}
where we write $\bar{p}_C(c) \overset{\Delta}{=} p_C(c)\dfrac{k^\alpha}{c^\alpha}$, which is bounded in the limit. For $\alpha \geq 1$ and $c \nsim k$, we have $c^\alpha/k^\alpha \rightarrow 0$.
\end{proof}
\begin{lemma}
\label{alpha1lemma}
For $\alpha=1$, $\tilde{p}_C(c)\overset{\Delta}{=} \dfrac{p_{C}(c) k}{c}$, and up to $o(1/k^2)$, we have the following lower bound for $\tilde{J}$:
\begin{equation}
\begin{aligned}
\hspace{-0.45cm}\tilde{J}_{\alpha=1}
\geq \dfrac{\pi^2 L}{k^22 \ln 2}\left[\int_0^1 \dfrac{dw}{f_W(w)w(1-w)}\right]^{-1}\left(1+L^{-1}\sum_{c \in w(1)\cap o(k)}\tilde{p}_{C}(c)+\dfrac{4L^{-1}}{\pi^2}\sum_{c \in O(1)}\tilde{p}_C(c)\right)
.
\end{aligned}
\end{equation}
\end{lemma}
\begin{proof}
For $\alpha =1$, all terms in $\tilde{J}$ become of order $1/k^2$, and we can write:
\begin{equation}
\label{ft3}
\begin{aligned}
\tilde{J}_{\alpha=1}=&\sum_{c \sim k}\tilde{p}_{C}(c)\left(\dfrac{1}{k^22 \ln 2 }\bar{\mathcal{J}}[g_c(W)]+o\left(\dfrac{1}{k^2}\right)\right)\\
+&\sum_{c \in w(1)\cap o(k)}\tilde{p}_{C}(c) \left(\dfrac{1}{k^22 \ln 2 }\bar{\mathcal{J}}[g_c(W)]+o\left(\dfrac{1}{k^2}\right)\right)
+\sum_{c \in O(1)}\tilde{p}_{C}(c)\dfrac{c^2}{k^2} \mathcal{I}_{c}(W,\bm{\pi}_{c})\\
&=\dfrac{1}{k^22 \ln 2} \Bigg(\sum_{c \sim k}\tilde{p}_{C}(c)\bar{\mathcal{J}}[g_c(W)]+\sum_{c \in w(1)\cap o(k)}\tilde{p}_{C}(c) \bar{\mathcal{J}}[g_c(W)]\\
&+\sum_{c \in O(1)}\tilde{p}_{C}(c)2 \ln 2\, c^2 \mathcal{I}_{c}(W,\bm{\pi}_{c})\Bigg)+o\left(\dfrac{1}{k^2}\right),
\end{aligned}
\end{equation}
where note that $\tilde{p}_C(c)=\dfrac{p_C(c)k}{c}$ is bounded in the limit. We can replace $\bar{\mathcal{J}}[g_c(W)]$ in the first two terms by $\bar{\mathcal{J}}[g(W)]$, and the remainder will still decay faster than $1/k^2$. That is, up to $o(1/k^2)$, we have the following:
\begin{equation}
\label{ineq}
\begin{aligned}
&\dfrac{1}{k^22 \ln 2} \left(\sum_{c \sim k}\tilde{p}_{C}(c)\bar{\mathcal{J}}[g(W)]+\sum_{c \in w(1)\cap o(k)}\tilde{p}_{C}(c) \bar{\mathcal{J}}[g(W)]+\sum_{c \in O(1)}\tilde{p}_{C}(c)2 \ln 2\, c^2 \mathcal{I}_{c}(W,\bm{\pi}_{c})\right)\\
&=\dfrac{1}{k^22 \ln 2} \left(\bar{\mathcal{J}}[g(W)]\sum_{c \in w(1)}\tilde{p}_{C}(c)+\sum_{c \in O(1)}\tilde{p}_{C}(c)2 \ln 2\, c^2 \mathcal{I}_{c}(W,\bm{\pi}_{c})\right)\geq\\
&\dfrac{1}{k^22 \ln 2} \left(\pi^2 \left[\int_0^1 \dfrac{dw}{f_W(w)w(1-w)}\right]^{-1}\sum_{c \in w(1)}\tilde{p}_{C}(c)+4\left[\int_0^1 \dfrac{dw}{f_W(w)w(1-w)}\right]^{-1}\sum_{c \in O(1)}\tilde{p}_{C}(c) \right)\\
&=\dfrac{1}{k^22 \ln 2}\left[\int_0^1 \dfrac{dw}{f_W(w)w(1-w)}\right]^{-1}\pi^2 \left(\sum_{c \in w(1)}\tilde{p}_{C}(c)+\dfrac{4}{\pi^2}\sum_{c \in O(1)}\tilde{p}_{C}(c) \right),
\end{aligned}
\end{equation}
where the inequality follows from {Lemma \ref{ineqlemma}} and {Lemma \ref{asymhm1}}. Finally, applying {Lemma \ref{lemmaL}}, we obtain the desired bound.
\end{proof}
\begin{lemma}
\label{alphag1lemma}
For $\alpha > 1$, we have the following lower bound for $\tilde{J}$:
\begin{equation}
\tilde{J}_{\alpha >1}\geq \dfrac{L}{k^22 \ln 2}\left[\int_0^1 \dfrac{dw}{f_W(w)w(1-w)}\right]^{-1}\pi^2+o\left(\dfrac{1}{k^2}\right),
\end{equation}
with conditions for equality being:
\begin{align}
&g(w)=g_{\mathrm{opt}}(w),\\
&\lim _{k \to \infty}p_C(k/L)=1,
\end{align}
where $g_{\mathrm{opt}}(w)$ is given by equation \eqref{gopt}.
\end{lemma}
\begin{proof}
Note that in equation \eqref{ft3}, for any choice of $\alpha >1$, the second two terms will decay faster than $1/k^2$, while the first term will dominate since it decays as $1/k^2$ for any value of $\alpha$. Therefore choosing $\alpha >1$ makes the contribution from $c\nsim k$ negligible in the limit, and so repeating the same steps as in the proof of {Lemma \ref{alpha1lemma}}, we obtain:
\begin{equation}
\begin{aligned}
&\tilde{J}_{\alpha>1}\geq\dfrac{1}{k^22 \ln 2}\left[\int_0^1 \dfrac{dw}{f_W(w)w(1-w)}\right]^{-1} \pi^2 \sum_{c \sim k}p_C(c)\dfrac{k}{c}+o\left(\dfrac{1}{k^2}\right)\\
&\geq \dfrac{L}{k^22 \ln 2}\left[\int_0^1 \dfrac{dw}{f_W(w)w(1-w)}\right]^{-1}\pi^2+o\left(\dfrac{1}{k^2}\right).
\end{aligned}
\end{equation}
Conditions for equality follow from {Lemma \ref{asymhm1}} and {Lemma \ref{convjen}}.
\end{proof}
\begin{lemma}
\label{asymjstar}
$\alpha>1$, $g_{\mathrm{opt}}(w)$ and $p^{*}_C(c)\overset{\Delta}{=}\delta(c,k/L)$ are asymptotically optimal for $\tilde{J}$.
\end{lemma}
\begin{proof}
This follows directly from {Lemma \ref{jstar}}, {Lemma \ref{alpha1lemma}}, and {Lemma \ref{alphag1lemma}}. Note that the lower bound for $\tilde{J}_{\alpha=1}$ given in {Lemma \ref{alpha1lemma}} is asymptotically larger than the lower bound for $\tilde{J}_{\alpha>1}$ given in {Lemma \ref{alphag1lemma}}, and since $g_{\mathrm{opt}}(w)$ and $p^{*}_C(c)$ achieve this lower bound for $\tilde{J}_{\alpha>1}$, we conclude that the optimal sequence of pirate strategies must have $\alpha >1$ with its limit being $(g_{\mathrm{opt}}(w), p^{*}_C)$.
\end{proof}
\begin{lemma}
\label{fstarlemma}
$\left[\mathlarger{\int}_0^1 \dfrac{dw}{f_W(w)w(1-w)}\right]^{-1}\pi^2 \leq 1$, with equality iff $f_W$ is the arcsine distribution: $f^*_W(w)\overset{\Delta}{=}\left(\pi \sqrt{w(1-w)}\right)^{-1}$.
\end{lemma}
\begin{proof}
This is \textit{Lemma 4} in \cite{HM2012TIFS}.
\end{proof}
\begin{lemma}
\label{f*opt}
$f^*_W$ is asymptotically optimal for $\tilde{J}$.
\end{lemma}
\begin{proof}
From {Lemma \ref{alphag1lemma}} and {Lemma \ref{asymjstar}}, the asymptotically optimal $f_W$ is the maximiser of $\dfrac{L}{k^22 \ln 2}\left[\mathlarger{\int}_0^1 \dfrac{dw}{f_W(w)w(1-w)}\right]^{-1}\pi^2$, which is $f^{*}_W$ by {Lemma \ref{fstarlemma}}.
\end{proof}
\begin{lemma}
\label{optlemma}
The payoff function $\tilde{J}_{k, L}$ has the following asymptotic behavior:
\begin{equation}
\max_{f_W}\min_{\{\pi\}, p_{{C}}}\tilde{J}_{k, L} \rightarrow \dfrac{L^2}{k^22 \ln 2},
\end{equation}
with asymptotically optimal strategies being $(f^{*}_W, g^{*}, p^{*}_{{C}})$, where $g^{*}(w) \overset{\Delta}{=} w$ is the interleaving attack.
\end{lemma}
\begin{proof}
This follows directly from {Lemma \ref{alphag1lemma}}, {Lemma \ref{asymjstar}}, and {Lemma \ref{f*opt}}.
\end{proof}
We are now in a position to present the main result of the work, concerning the asymptotic behavior of our payoff function as defined in equation \eqref{payoff}. To this end, we define the following degenerate distribution:
\begin{equation}
p^{*}_{\hat{C}}(\hat{c})\overset{\Delta}{=}\prod_{l=1}^{L}\delta(c^l,k/L)=\begin{cases}1, &\hat{c}= (k/L, k/L, \dots, k/L)\\ 0, &\mathrm{otherwise}\end{cases},
\end{equation}
which corresponds to the strategy where the pirates equally populate the $L$ TV channels.
\begin{theorem}
\label{th:main}
The payoff function in equation \eqref{payoff} has the following asymptotic behavior:
\begin{equation}\max_{f_W}\min_{\{\pi\}, p_{\hat{C}}}{J}_{k, L} \rightarrow \dfrac{L^2}{k^22 \ln 2},\end{equation} with asymptotically optimal strategies being $(f^{*}_W, g^{*}, p^{*}_{\hat{C}})$.
\end{theorem}
\begin{proof}
From equation \eqref{payoff}, we can see that $J=L\tilde{J}+k^{-1}I(\hat{Y}; \hat{C}|\hat{W})$, and the second term has zero as its minimum. For any $f_W$ and $g$, this is achieved when $\hat{C}$ is deterministic, i.e. as in $p^{*}_{\hat{C}}$. Optimality of $(f^{*}_W, g^{*}, p^{*}_{\hat{C}})$ follows from {Lemma \ref{optlemma}}, as well as the asymptotic value.
\end{proof}


\section{Alternative payoff and fingerprinting capacity}
\label{sec:alt}
The maximin game of our payoff given by equation \eqref{multipayoff} (or equation \eqref{payoff} for the binary alphabet case) is a direct generalization of the single TV channel payoff in \cite{HM2012TIFS}, i.e. if we set $L=1$ for the maximin game, we recover the single TV channel payoff that defines the fingerprinting capacity, which is asymptotically equal to $\big(k^22 \ln 2\big)^{-1}$ in the binary alphabet case \cite{HM2012TIFS}. Nevertheless, this is not the only generalization that reduces to the single TV channel payoff when setting $L=1$. If we consider the payoff $R_{k, L}\overset{\Delta}{=}k^{-1}I(\hat Y; \hat X_{\cal K}| \hat {\mathbf{W}})$, the maximin game will also reduce to the single TV channel case when $L=1$. This payoff is generally smaller than our payoff $J_{k, L}$, and the difference term (up to $k$) is $I(\hat Y; \hat S|\hat{X}_{\mathcal{K}})=I(\hat Y; \hat{\mathbf{M}}|\hat{X}_{\mathcal{K}})$, i.e. we can write:
\begin{equation}
I(\hat Y; \hat X_{\cal K}| \hat {\mathbf{W}})=k J_{k, L}-I(\hat Y; \hat S|\hat{X}_{\mathcal{K}})=k J_{k, L}-I(\hat Y; \hat{\mathbf{M}}|\hat{X}_{\mathcal{K}})=I(\hat Y; \hat{\mathbf{M}}|\hat{\mathbf{W}})-I(\hat Y; \hat{\mathbf{M}}|\hat{X}_{\mathcal{K}}).
\end{equation}
$R_{k, L}$ could possibly have different asymptotic behavior compared to $J_{k, L}$. However, since it can not be larger than $J_{k, L}$, we can write (in the binary alphabet case):
\begin{equation}
\max_{f_W}\min_{\{\pi\}, p_{\hat{C}}}{R}_{k, L} \rightarrow \dfrac{N^2}{k^22 \ln 2},
\end{equation}
for some $N \leq L$. It is plausible that $N$ is in fact equal to $L$, and the asymptotic maximin games for both $R_{k, L}$ and $J_{k, L}$ are identical. This is motivated by the fact that the difference term $I(\hat Y; \hat S|\hat{X}_{\mathcal{K}})$ extracts information about the choice of pirates assigned to the different TV channels from their outputs $\hat{Y}$, and asymptotically we expect likely realizations of $\hat{X}_{\mathcal{K}}$ to be insensitive to who the pirates choose to assign for the different TV channels. That is, as long as they choose large enough numbers of pirates to populate the TV channels, we expect that knowing $\hat{S}$ will not be significant if we already know $\hat{X}_{\mathcal{K}}$. This of course hinges on the argument that the pirates should assign large numbers to all $L$ TV channels, which is also in line of what one would expect. This leads us to the following conjecture:
\begin{conjecture}
\label{theconj0}
For the binary alphabet case, and $\tilde{L}=1$, the payoff functional $R_{k, L}$ given by:
\begin{equation}
R_{k, L}=J_{k, L}-\dfrac{1}{k}I(\hat Y; \hat S|\hat{X}_{\mathcal{K}})
\end{equation}
has the following asymptotic behavior:
\begin{equation}\max_{f_W}\min_{\{\pi\}, p_{\hat{C}}}{R}_{k, L} \rightarrow \dfrac{L^2}{k^22 \ln 2},\end{equation} with optimal strategies $f_W^*(w)=\left(\pi \sqrt{w(1-w)}\right)^{-1}$, $g^*(w)=w$, and $p^{*}_{\hat{C}}(\hat{c})=\prod_{l=1}^{L}\delta(c^l,k/L)$.
\end{conjecture}
The payoffs $J_{k, L}$ and $R_{k, L}$ are the most natural generalizations of the single TV channel payoff in \cite{HM2012TIFS}. Assuming Conjecture \ref{theconj0} holds, then by Theorem \ref{th:main}, and Corollary 7 in \cite{HM2012TIFS}, the asymptotic maximin value for both $J_{k, L}$ and $R_{k, L}$ is the same as the asymptotic fingerprinting capacity in the single TV channel case ($L=1$). This prompts the question of whether this holds also for any $L$, i.e. if our payoff defines the fingerprinting capacity for multiple TV channels. Since the pirates choose $\hat{S}$ independently of the watermarking procedure, and once $\hat{S}$ is known the problem reduces to $L$ independent single TV channels that abide by the assumptions used in \cite{HM2012TIFS}, our direct generalization of the single TV channel payoff can very possibly define the capacity (in the limit $k \to \infty$) in this case as its maximin value, and this leads us to the following conjecture:
\begin{conjecture}
\label{theconj}
The binary fingerprinting capacity in the multiple TV channels scenario (with $\tilde{L}=1$) $C_{\rm fp}^{\rm binary}(k,L)$ has the asymptotic behavior:
\begin{equation}
 \max_{f_W}\min_{\{\pi\}, p_{\hat{C}}}C_{\rm fp}^{\rm binary}(k,L) \rightarrow \dfrac{L^2}{k^22 \ln 2}
  \end{equation}
with optimal strategies $f_W^*(w)=\left(\pi \sqrt{w(1-w)}\right)^{-1}$, $g^*(w)=w$, and $p^{*}_{\hat{C}}(\hat{c})=\prod_{l=1}^{L}\delta(c^l,k/L)$.
\end{conjecture}

\section{Discussion}
\label{sec:discussion}
We have shown that when $k$ colluders attack $L$ TV channels simultaneously, the maximin value of the payoff defined in equation \eqref{payoff} has the asymptotic value $\big((k/L)^22 \ln 2\big)^{-1}$ with asymptotically optimal strategies being $f_W^*(w)=\left(\pi \sqrt{w(1-w)}\right)^{-1}$, $g^*(w)=w$, and $p^{*}_{\hat{C}}(\hat{c})=\prod_{l=1}^{L}\delta(c^l,k/L)$. This is exactly the same as having $L$ independent single TV channels \cite{HM2012TIFS} with $k/L$ colluders on each TV channel. This is expected as the payoff is independent of the identity of the pirates assigned to the TV channels, i.e. it is only sensitive to the number of colluders assigned to the different TV channels. The identity dependence is present in $I(\hat Y; \hat X_{\cal K}| \hat W)$, namely in the difference term $I(\hat Y; \hat S|\hat{X}_{\mathcal{K}})$, which possibly lowers the value of the payoff by virtue of allowing the identity of the pirates assigned to a TV channel to be chosen uniformly randomly given they know the numbers $\{c^l\}_{l=1}^L$.

We have presented arguments that strongly suggest that this liberty the pirates enjoy, although it can be significantly advantageous for finite $k$, does become of less value asymptotically and the asymptotic maximin games of $I(\hat Y; \hat X_{\cal K}| \hat W)$ and $I(\hat Y; \hat{Z}, \hat{C}| \hat W)$ are the same. It remains an open question to prove Conjecture \ref{theconj0}, which is left for future work. Furthermore, we have argued that our payoff does predict how the fingerprinting capacity behaves asymptotically, and that it is indeed plausible that its maximin value can be replaced by the fingerprinting capacity in Theorem \ref{th:main}.

It is important to note that we have not used any saddle point property of the payoff (c.f. \cite{HM2012TIFS}) as it is simply not needed. That is, we make no claims of the maximin game yielding the same value for the payoff as the minimax game (applicability of Sion's theorem) for any of the payoffs discussed in this work. Although this is very possible, i.e. the strategies we found could be an equilibrium point, one still needs to show that the payoff considered is jointly convex in the pirates' strategies, e.g. if future directions require solving the maximin game numerically.

As other potential follow-up work, there remains the generalization to non-binary alphabet, the case where users can tune in to more than one TV channel simultaneously, i.e. $\tilde{L} \neq 1$, dynamic tracing, and perhaps most practically useful, the study of decoders and score functions suitable for multiple TV channel attacks.

\subsection*{Acknowledgements}
Part of this work was supported by NWO grant CS.001 (Forwardt).


\bibliographystyle{plain}
\bibliography{references}



\end{document}